\documentclass[9pt,a4paper,twocolumn,twoside]{extarticle}
\usepackage[english]{babel}

\usepackage[T1]{fontenc}
\usepackage[notextcomp]{stix2}
\usepackage{FiraSans}
\usepackage{amsmath,amsthm}
\usepackage[
    left=1.25cm,
    right=1.25cm,
    top=2cm,
    bottom=2cm,
    headsep=0.75cm
]{geometry}
\usepackage{xcolor}
\definecolor{rhocolor}{rgb}{0.12,0.3,0.17}
\usepackage{microtype}
\usepackage{authblk}
\usepackage{enumitem}
\usepackage{fancyhdr}
\usepackage{lastpage}
\usepackage{lettrine}
\usepackage[explicit]{titlesec}
\usepackage{xurl}
\usepackage[colorlinks=true,allcolors=rhocolor]{hyperref}
\usepackage{cite}
\newcommand{\printbibliography}{%
  \bibliographystyle{IEEEtran}%
  \bibliography{references}%
}

\setlist{noitemsep,topsep=2pt}

\AtBeginDocument{
  \setlength{\abovedisplayskip}{10pt plus 2pt minus 2pt}
  \setlength{\belowdisplayskip}{10pt plus 2pt minus 2pt}
  \setlength{\abovedisplayshortskip}{8pt plus 2pt minus 2pt}
  \setlength{\belowdisplayshortskip}{8pt plus 2pt minus 2pt}
}

\makeatletter
\newcommand{\doctype}[1]{\def\@doctype{#1}}
\newcommand{\dates}[1]{\def\@dates{#1}}
\newcommand{\keywords}[1]{\def\@keywords{#1}}
\providecommand{\keywordname}{Keywords:}
\newcommand{\corres}[1]{\g@addto@macro\@corres{\par #1\par}}
\let\@corres\@empty

\newcommand{\absfont}{\normalfont\sffamily\fontsize{8.5}{11}\selectfont}
\newcommand{\absheadfont}{\color{rhocolor}\normalfont\sffamily\fontsize{9}{11}\selectfont\bfseries}
\newcommand{\keywordsfont}{\normalfont\sffamily\itshape\fontsize{7.8}{9}\selectfont}
\newcommand{\keywordheadfont}{\normalfont\sffamily\fontsize{7.8}{9}\selectfont\bfseries}
\newcommand{\theabstract}{}
\def\xabstract{abstract}
\long\def\abstract#1\end#2{%
    \def\two{#2}\ifx\two\xabstract
    \long\gdef\theabstract{\ignorespaces#1}%
    \def\go{\end{abstract}}%
\else
    #1\end{#2}%
    \let\go\relax\fi
\go}

\newcommand{\rhoabstract}{%
  \begingroup
  \setlength{\fboxsep}{6pt}%
  \noindent\colorbox{rhocolor!13}{%
    \parbox{\dimexpr\textwidth-2\fboxsep\relax}{%
      \absheadfont\abstractname\par\vspace{0.45em}%
      \absfont\theabstract\par
      \ifx\@keywords\undefined\else
        \vspace{0.7em}%
        {\keywordheadfont\keywordname\ }{\keywordsfont\@keywords}%
      \fi
    }%
  }%
  \endgroup
}

\newcommand{\titlefont}{\bfseries\color{rhocolor}\fontsize{18}{22}\sffamily\selectfont}
\renewcommand{\@maketitle}{%
  \vspace*{-1.5em}%
  {\raggedright\sffamily\bfseries\ifx\@doctype\undefined\else\@doctype\par\fi}%
  \vspace{0.35em}%
  {\raggedright\titlefont\@title\par}%
  \vspace{0.5em}%
  {\raggedright\normalfont\sffamily\@author\par}%
  \ifx\@dates\undefined\else
    \vspace{0.7em}{\raggedright\sffamily\fontsize{7}{8}\selectfont\@dates\par}%
  \fi
  \vspace{1.2em}%
  \rhoabstract
  \ifx\@corres\@empty\else
    \vspace{0.8em}{\footnotesize\sffamily\@corres}%
  \fi
  \vspace{1.3em}%
}

\newcommand{\footerfont}{\normalfont\sffamily\fontsize{8}{9}\selectfont}
\fancypagestyle{firststyle}{%
  \fancyhf{}%
  \fancyfoot[R]{\footerfont\textbf{\thepage}\textendash\pageref{LastPage}}%
}

\newcommand{\dropcapfont}{\color{rhocolor}\bfseries}
\newcommand{\rhostart}[1]{\lettrine[lines=2,lraise=0,findent=2pt,nindent=0em]{{\fontsize{27}{24}\dropcapfont #1}}{}}

\titleformat{\section}
    {\color{rhocolor}\sffamily\large\bfseries}
    {\thesection.}
    {0.5em}
    {\MakeUppercase{#1}}
    []
\titleformat{name=\section,numberless}[block]
    {\color{rhocolor}\sffamily\large\bfseries}
    {}
    {0em}
    {\rule{1.25ex}{1.25ex}\hspace{2.5pt}\MakeUppercase{#1}}
    []
\titleformat{\subsection}[block]
    {\sffamily\bfseries}
    {\thesubsection.}
    {0.5em}
    {#1}
    []
\titleformat{\subsubsection}[block]
    {\small\sffamily\bfseries\itshape}
    {\thesubsubsection.}
    {0.5em}
    {#1}
    []
\titlespacing*{\section}{0pc}{3ex plus 4pt minus 3pt}{5pt}
\makeatother

\newtheorem{theorem}{Theorem}[subsection]
\newtheorem{proposition}[theorem]{Proposition}

\theoremstyle{definition}
\newtheorem{definition}[theorem]{Definition}

\doctype{Preprint}
\title{A Probabilistic Interpretation of the Ball Mapper Graph}

\author[{\textasteriskcentered},a, b]{John Rick Manzanares}
\author[c]{Jay-Anne Bulauan}

\affil[a]{Dioscuri Centre in Topological Data Analysis, Institute of Mathematics of the Polish Academy of Sciences, 00-656 Warsaw, Poland}
\affil[b]{International Environmental Doctoral School, University of Silesia in Katowice, 41-205 Sosnowiec, Poland}
\affil[c]{Independent Researcher}

\dates{\today}

\corres{\textsuperscript{\textasteriskcentered}Corresponding author: \href{mailto:jdolormanzanares@impan.pl}{jdolormanzanares@impan.pl}}

\begin{abstract}
We introduce Probabilistic Ball Mapper, a formulation of Ball Mapper in which each data point is assigned a probability distribution supported only on the metric balls that contain it. This assignment defines both a partition subordinate to the Ball Mapper cover and a Markov kernel from the finite data space to the cover. We study two assignment schemes: a uniform-on-support rule and a localized radial-basis rule that incorporates distance to landmarks while preserving the underlying cover. Pushing the empirical data distribution through the kernel produces a probability distribution over vertices. Drawing twice, conditionally and independently, from each pointwise distribution produces a soft-overlap matrix. This matrix is symmetric, nonnegative, positive semidefinite, and has the vertex distribution as both marginals. It therefore provides a mass-normalized refinement of classical Ball Mapper overlap rather than another unnormalized edge count. For graphs constructed on a common cover, the vertex and overlap distributions can be compared directly. For independently fitted covers, we formulate Wasserstein and fused Gromov--Wasserstein-type discrepancies that account for vertex mass, landmark geometry when a common ambient metric is available, and intrinsic graph relations. For a fixed cover, we derive explicit perturbation bounds controlled by the sensitivity of the assignment rule, the magnitude of the data perturbation, and the data mass near cover boundaries. When the cover is recomputed, landmark motion creates an additional source of variation, for which we state a transport-based stability principle rather than an unconditional theorem. The resulting framework turns Ball Mapper into a probability-valued representation suitable for quantitative comparison while retaining its geometric interpretability and computational simplicity.
\end{abstract}

\keywords{Markov kernel; soft assignment; optimal transport; graph comparison; topological data analysis.}

\begin{document}

\maketitle
\thispagestyle{firststyle}

\section{Introduction}

This section introduces the problem addressed by Probabilistic Ball Mapper and situates the proposed framework within existing work on Ball Mapper and related Mapper constructions. We first motivate the use of probability-valued point-to-cover assignments, then discuss the closest related approaches and summarize the principal contributions of this work.

\subsection{Motivation}

\rhostart{T}opological data analysis (TDA) studies the shape of data through combinatorial and geometric summaries. A recurring theme in geometric data analysis is that high-dimensional observations may concentrate near lower-dimensional geometric structures, so that distances, neighborhoods, and connectivity can reveal organization that is not apparent from coordinatewise summaries alone \cite{tenenbaum2000global,roweis2000nonlinear}. Mapper \cite{singh2007mapper} constructs a simplicial summary of data by combining a filter function, a cover of the filter space, and clustering in the preimages of cover elements. Persistent homology provides another major class of TDA summaries by tracking the appearance and disappearance of topological features across a filtration \cite{edelsbrunner2002topological}. A key reason for its usefulness in data analysis is the existence of stability theorems, which show that small perturbations of suitable input functions lead to controlled changes in the resulting persistence diagrams \cite{cohensteiner2007stability, chazal2016persistence}.

Ball Mapper \cite{dlotko2019ballmapper} is a Mapper-inspired method that represents a finite metric dataset by covering it with balls centered at selected landmarks. The resulting graph is the $1$-skeleton of the nerve of this cover. Ball Mapper is well suited to exploratory data analysis because it uses the metric structure of the original data directly and requires only a scale parameter, the ball radius. Vertices represent metric balls, edges record non-empty overlaps, and scalar quantities may be visualized by coloring the vertices. 

The incidence relation underlying classical Ball Mapper is binary: a point either belongs to a covering ball or it does not. The original construction also allows edges to be weighted and describes a filtered nerve in which a simplex is weighted by the number of data points simultaneously covered by its vertices \cite{dlotko2019ballmapper}. These count-based weights quantify the size of an overlap, but they do not distinguish the relative proximity of a point to the landmarks of the balls that contain it. Thus, when several balls overlap, a point close to one landmark and near the boundary of another contributes equally to both incidences. Moreover, the resulting counts are not the marginals of a point-to-cover probability law, which makes it difficult to compare covers with different numbers of points or vertices in a coherent mass-preserving way. A small perturbation near a ball boundary may also create or destroy an incidence, and hence an edge, abruptly.

In this study, we formulate Ball Mapper probabilistically. Instead of assigning a point only by set membership, we assign to each point a probability measure on the covering balls. This produces a Markov kernel from the finite data space to the finite cover. We consider two choices: a uniform measure on the balls containing the point, and a localized radial-basis weighting that gives greater mass to balls whose landmarks are closer to the point. Both choices are supported only on the balls that contain the point, so the probabilistic construction remains subordinate to the original Ball Mapper cover. 

The probabilistic viewpoint gives more than a soft visualization. It induces a marginal probability distribution on the vertices of the Ball Mapper graph, which can be interpreted as a density-aware vertex weighting. It also defines a soft overlap matrix measuring how much probability mass is jointly assigned to pairs of balls. This matrix provides a quantitative analogue of adjacency that varies more continuously under perturbations than the binary edge set. 

We then use this structure to compare probabilistic Ball Mapper graphs. If two datasets are projected onto the same cover, their vertex weights and soft overlap matrices can be compared directly. More generally, two independently constructed Ball Mapper graphs may have different numbers of vertices. To handle this case, we introduce an optimal transport-based pseudometric. A transport operator maps the vertex distribution of one probabilistic Ball Mapper graph to that of another, while minimizing a cost based on landmark geometry and local graph descriptors. This connects probabilistic Ball Mapper to the broader optimal transport and Gromov--Wasserstein perspectives on comparing weighted metric and graph-like objects. 

Finally, we discuss stability. Away from cover boundaries, localized probabilistic memberships vary smoothly with the data. Consequently, vertex weights and soft overlaps change continuously under small perturbations. The main source of instability is the same one present in classical Ball Mapper: points near ball boundaries may change their support. The resulting stability statement therefore contains a perturbation term together with a boundary-mass term. This gives a natural explanation of when probabilistic Ball Mapper graph comparisons are expected to be stable and when abrupt changes may occur. 

\subsection{Related Work}

Several Mapper variants use graded or probabilistic membership, but they use it for purposes different from the construction developed here. F-Mapper applies fuzzy $C$-means to filter values to obtain irregular overlapping cover elements; a membership threshold determines the point-to-cover incidences from which the usual Mapper graph is built \cite{bui2020fmapper}. Soft Mapper instead models cover incidences by conditional Bernoulli variables and thereby defines a probability distribution over Mapper graphs, primarily to make filter optimization differentiable \cite{oulhaj2024differentiablemapper}. The Gaussian mixture formulation of Tao and Ge uses mixture responsibilities and a latent assignment matrix to construct implicit intervals, and summarizes the resulting graph distribution by a Mapper graph mode \cite{tao2025implicitmapper}. These approaches remain filter-based and use soft assignments to construct, randomize, or optimize the cover and its resulting combinatorial graph.

The present construction has a different starting point and output. We retain the metric ball cover of Ball Mapper and require the probability assignment to be subordinate to that fixed cover. The entire membership vector is retained as part of the representation, rather than thresholded into incidence. The result is a deterministic probability-enriched graph described by the point-to-cover kernel, its pushforward vertex distribution $\nu$, and its joint overlap distribution $Q$, rather than a distribution over combinatorial graphs. In the language of cover learning, the membership map $w\colon X\to\Delta^{K-1}$ is a partition of unity subordinate to the Ball Mapper cover; unlike methods that learn a cover such as ShapeDiscover \cite{scoccola2025coverlearning}, the emphasis of this study is on the probability objects induced by a prescribed metric-ball cover and on comparing those objects.

Two further uses of the word ``stability'' or ``probabilistic'' should be distinguished. GK-Mapper studies the response of fuzzy Mapper graphs to the fuzzifier parameter and identifies threshold events at which the graph changes \cite{sen2026gkmapper}. Our proven bounds instead concern perturbations of the data under a fixed Ball Mapper cover. Separately, probabilistic convergence of random Mapper graphs concerns recovery from random samples and convergence to a Reeb-type object \cite{brown2021randommapper}; it does not define a point-to-cover Markov kernel of the kind used here.

The main contributions of this work are as follows:
\begin{enumerate}
  \item We formulate point-to-cover assignment as a Markov kernel, or equivalently a subordinate partition of unity, and give both a uniform-on-support rule and localized radial-basis rules.
  \item We derive the induced vertex distribution $\nu$ and the soft overlap matrix $Q=W^{\mathsf T}D_{\mu}W$.  We show that $Q$ is a joint probability distribution with marginals $\nu$ and explain precisely how it differs from the overlap-count weights already available in the classical Ball Mapper.
  \item We give direct comparisons for graphs on a common ordered cover and transport-based discrepancies for independently fitted covers with different vertex sets. The latter combine vertex features with relational graph information without requiring a one-to-one vertex correspondence.
  \item We prove explicit perturbation bounds for $\nu$ and $Q$ under a fixed cover, with a separate term for data mass near cover boundaries. For recomputed covers, we isolate landmark motion as an additional source of variation and formulate an associated stability principle.
\end{enumerate}

\subsection{Organization of the Paper}

The paper is organized as follows. Section~2 recalls the Ball Mapper construction, defines subordinate probabilistic assignments, develops the Markov-kernel formulation, and derives the vertex distribution and soft overlap matrix.  A minimal worked example makes the normalization and marginal relations explicit. Section~3 develops direct common-cover comparisons and transport-based comparisons for independently fitted covers, and then proves the fixed-cover perturbation bounds while discussing the additional difficulty created by recomputing the landmarks.  Section~4 summarizes the scope, limitations, and directions for further work.

\section{Probabilistic Ball Mapper}

Ball Mapper turns a finite metric dataset into a graph by covering the dataset with metric balls and recording how these balls overlap. Each vertex of the graph represents one covering ball, and an edge indicates that two balls share at least one data point. Thus, the graph gives a coarse, scale-dependent summary of the shape of the data.

In the classical construction, point-to-ball membership is binary: a data point either lies in a ball or it does not. In this section, we first recall the classical Ball Mapper construction. We then replace binary membership by probability measures on the cover. This produces a soft assignment of points to covering balls, equivalently a Markov kernel from the dataset to the set of balls, and leads naturally to density-aware weights on the Ball Mapper graph.

\subsection{Metric Structure}

Let $(X,d)$ be a finite metric space. We think of $X$ as a finite dataset and of $d$ as the rule that determines which data points are close. The function $d : X \times X \to \mathbb{R}$ satisfies, for all $x,y,z \in X$,
\begin{itemize}
\item $d(x,y) \ge 0$,
\item $d(x,y) = 0$ if and only if $x = y$,
\item $d(x,y) = d(y,x)$, and
\item $d(x,z) \le d(x,y) + d(y,z)$.
\end{itemize}
The metric is therefore the geometric input of Ball Mapper: it determines the balls used to cover the dataset.

Fix a scale parameter $\varepsilon>0$. This parameter controls the resolution of the construction. Small values of $\varepsilon$ produce a finer cover, while larger values produce a coarser one. To build the cover, we first choose a set of representative data points.

\begin{definition}
Let $(X,d)$ be a metric space and let $\varepsilon > 0$. A subset $N \subset X$ is an \emph{$\varepsilon$-net} if
\begin{itemize}
  \item for every $x \in X$, there exists $n \in N$ such that $d(x,n) < \varepsilon$; and
  \item for any distinct $n_1,n_2 \in N$, $d(n_1,n_2) \ge \varepsilon$.
\end{itemize}
The elements of $N$ are called \emph{landmarks}.
\end{definition}

Let $N_\varepsilon=\{c_1,\dots,c_K\}$ be an $\varepsilon$-net in $X$. Each landmark $c_i$ determines the metric ball
\[
C_i = B(c_i,\varepsilon) = \{y\in X \mid d(c_i,y)<\varepsilon\}.
\]
The collection $\mathcal C_\varepsilon(X)=\{C_1,\dots,C_K\}$ is a cover of $X$. A \emph{cover} of a metric space $X$ is a collection $\{U_\alpha\}_{\alpha \in I}$ of subsets of $X$ such that
\[
\bigcup_{\alpha \in I}U_\alpha = X.
\]
By the defining property of an $\varepsilon$-net, every point of $X$ lies within distance $\varepsilon$ of at least one landmark.

The combinatorial structure of the cover is encoded by its nerve \cite{hatcher2002algebraic}. Recall that an \emph{abstract simplicial complex} consists of a vertex set $V$ together with a collection $S$ of non-empty finite subsets of $V$ such that
  \begin{itemize}
  \item every singleton $\{v\}$ belongs to $S$, and
  \item whenever $\sigma \in S$ and $\tau \subset \sigma$ is non-empty, then $\tau \in S$.
  \end{itemize}
An element $\sigma \in S$ is a $k$-simplex when $k = |\sigma|-1$. The \emph{$k$-skeleton} of a simplicial complex is the subcomplex consisting of all simplices of dimension at most $k$.

\begin{definition}
Let $U=\{U_i\}_{i\in I}$ be a finite cover. The \emph{nerve} of $U$ is the abstract simplicial complex $\mathcal N(U)$ with vertex set $I$ such that a finite subset $\{i_0,\dots,i_m\} \subset I$ forms an $m$-simplex whenever
\[
\bigcap_{j=0}^m U_{i_j} \neq \emptyset.
\]
\end{definition}

The \emph{Ball Mapper graph} associated with the cover $\mathcal C_\varepsilon(X)$ is the $1$-skeleton of the nerve $\mathcal N(\mathcal C_\varepsilon(X))$. Equivalently, it is the graph with one vertex for each ball $C_i$ and an edge between two distinct vertices $C_i$ and $C_j$ precisely when
\[
C_i\cap C_j\neq \emptyset.
\]
Thus, two covering balls are adjacent exactly when they share at least one data point.

The classical Ball Mapper graph records overlap using only binary membership information: a point either contributes to a ball or it does not. The probabilistic construction introduced next keeps the same cover and the same underlying metric scale, but replaces this hard membership relation by a probability distribution over the balls associated with each data point.

\subsection{Probabilistic Modelling}

We now replace the hard membership relation between data points and covering balls by a probabilistic one. The guiding idea is that, in classical Ball Mapper, if a point $x$ lies in several covering balls, then it belongs to all of them equally in a binary sense. In the probabilistic version, we instead allow $x$ to distribute one unit of mass among the balls that contain it.

Since the cover $\mathcal C_\varepsilon(X)$ contains only finitely many balls, a probability measure on the cover is simply a way of assigning nonnegative weights to these balls so that the total weight is one. We write $2^A$ for the collection of all subsets of a finite set $A$. In this finite setting, there is no measurability issue as every subset of $A$ is allowed to have a probability.

\begin{definition}
Let $A$ be a finite set. A \emph{probability measure} on $A$ is a function $\mathbb P:2^A\to[0,1]$ such that $\mathbb P(A)=1$ and, for any pairwise disjoint subsets $E_1,\dots,E_m\subseteq A$,
\[
\mathbb P\left(\bigcup_{\ell=1}^m E_\ell\right) = \sum_{\ell=1}^m \mathbb P(E_\ell).
\]
\end{definition}

Equivalently, a probability measure on $A=\{a_1,\dots,a_m\}$ is determined by nonnegative numbers $p_1,\dots,p_m\ge 0$ satisfying
\[
\sum_{\ell=1}^m p_\ell=1,
\]
where $p_\ell=\mathbb P({a_\ell})$.

For $K\ge 1$, the \emph{probability simplex} on $K$ elements is
\[
\Delta^{K-1} =  \left\{p=(p_1,\dots,p_K)\in\mathbb R^K \mid p_i\ge 0 \text{ and } \sum_{i=1}^K p_i=1 \right\}.
\]
A probability measure on $\mathcal C_\varepsilon(X)=\{C_1,\dots,C_K\}$ can be identified with a vector in $\Delta^{K-1}$.

Let $\mathcal C_\varepsilon(X) = \{C_1,\dots,C_K\}$ be the Ball Mapper cover, where $C_i=B(c_i,\varepsilon)$. For each $x\in X$, we define the \emph{index set}
\[
I_x= \{i \in \{1,\dots,K\} \mid x\in C_i\}.
\]
The function $I_x$ records the covering balls that contain $x$. Since $\mathcal C_\varepsilon(X)$ covers $X$, we have $I_x \neq \emptyset$ for every $x \in X$.

A probabilistic Ball Mapper assignment associates to each data point $x\in X$ a probability measure $\kappa_x$ on the finite set of covering balls $\mathcal C_\varepsilon(X)$. The value
\[
\kappa_x({C_i})
\]
is interpreted as the amount of mass that $x$ assigns to the ball $C_i$.

We require this assignment to be \emph{subordinate} to the Ball Mapper cover. That is, a point may assign positive mass only to balls that contain it. In other words,
\[
\kappa_x({C_i})=0
\]
whenever $x \notin C_i$. Equivalently, the support 
\[
\{C_i\in\mathcal C_\varepsilon(X) \mid \kappa_x({C_i})>0\}
\]
of $\kappa_x$ is contained in $\{C_i \mid i\in I_x\}$.

A convenient way to construct such measures is by assigning nonnegative \emph{scores} (or real numbers) to the balls and then normalizing them. For each $x\in X$, let
\[
a_i(x) \ge 0, \qquad i=1,\dots,K,
\]
be scores satisfying $a_i(x)=0$ if $i \notin I_x$, and
\[
\sum_{j=1}^K a_j(x) > 0.
\]
The second condition holds, for example, whenever at least one ball containing $x$ receives positive score. We then define, for every subset $A\subseteq \mathcal C_\varepsilon(X)$,
\[
\kappa_x(A) = \frac{\sum_{\{i \mid C_i\in A\}} a_i(x)}{\sum_{j=1}^K a_j(x)}.
\]
Equivalently, on singletons,
\[
\kappa_x({C_i}) = \frac{a_i(x)}{\sum_{j=1}^K a_j(x)}.
\]
Because the scores vanish outside $I_x$, this probability measure is automatically subordinate to the cover.

\begin{proposition}
For every $x\in X$, the function $\kappa_x$ defined above is a probability measure on $\mathcal C_\varepsilon(X)$.
\end{proposition}

\begin{proof}
Fix $x\in X$ and write
\[
Z_x=\sum_{j=1}^K a_j(x).
\]
By assumption, $Z_x > 0$. Since each $a_i(x)$ is nonnegative, we have
\[
0\le \kappa_x(A)\le 1
\]
for every $A\subseteq\mathcal C_\varepsilon(X)$. Also, $\kappa_x(\emptyset)=0$ and
\[
\kappa_x(\mathcal C_\varepsilon(X)) = \frac{\sum_{i=1}^K a_i(x)}{Z_x} = 1.
\]

Now let $A_1,\dots,A_m\subseteq \mathcal C_\varepsilon(X)$ be pairwise disjoint. Then
\[
\sum_{\left\{i \ \mid \ C_i\in \bigcup_{\ell=1}^m A_\ell\right\}} a_i(x) =\sum_{\ell=1}^m \sum_{\{i \ \mid \ C_i\in A_\ell\}} a_i(x)
\]
because no ball belongs to more than one of the sets $A_\ell$. Dividing by $Z_x$ gives
\[
\kappa_x\left(\bigcup_{\ell=1}^m A_\ell\right) = \sum_{\ell=1}^m \kappa_x(A_\ell).
\]
Thus, $\kappa_x$ is a probability measure on $\mathcal C_\varepsilon(X)$.
\end{proof}

This general construction contains several useful choices. First, the uniform-on-support, or normalized-incidence, rule is obtained by taking
\[
a_i(x)=
\begin{cases}
1, & i\in I_x,\\
0, & i\notin I_x.
\end{cases}
\]
Then
\[
\kappa_x({C_i}) =
\begin{cases}
\dfrac{1}{|I_x|}, & i\in I_x,\\[0.5em] 
0, & i\notin I_x.
\end{cases}
\]
Here, $x$ distributes its mass equally among all covering balls that contain it. This choice is purely combinatorial as it depends only on the incidence relation $x \in C_i$.

More generally, we may prescribe any discrete probability distribution on the balls containing $x$. That is, for each $x\in X$, choose numbers $q_i(x) \ge 0$ where $i \in I_x$ satisfying
\[
\sum_{i\in I_x}q_i(x)=1.
\]
We then define
\[
\kappa_x({C_i}) =
\begin{cases}
q_i(x), & i\in I_x,\\
0, & i\notin I_x.
\end{cases}
\]
This allows the model to include prior information, local density information, confidence scores, or application-specific preferences among the balls that contain $x$.

A geometrically natural choice is obtained from a radial profile \cite{buhmann2003radial}. Let
\[
\rho:[0,\infty)\to[0,\infty)
\]
be a nonnegative function such that $\rho(r) > 0$ for $0 \le r <\varepsilon$. We assign scores
\[
a_i(x) =
\begin{cases}
\rho(d(x,c_i)), & i\in I_x,\\
0, & i\notin I_x.
\end{cases}
\]
Then
\[
\kappa_x({C_i}) =
\begin{cases}
\dfrac{\rho(d(x,c_i))}{\sum_{j\in I_x}\rho(d(x,c_j))}, & i\in I_x,\\[1.2em]
0, & i\notin I_x.
\end{cases}
\]
This produces a localized radial basis assignment. It means that the mass of $x$ is restricted to balls that contain $x$, but among those balls the mass can depend on the distance from $x$ to the corresponding landmarks.

For example, the Gaussian radial profile is
\[
\rho(r)=\exp(-\gamma r^2), \qquad \gamma>0.
\]
In Euclidean settings, we often write
\[
\gamma=\frac{1}{2\sigma^2}, \qquad \sigma>0,
\]
where $\sigma$ controls the radial scale. Other choices are also possible, such as
\[
\rho(r)=1,
\]
which recovers the uniform assignment, or
\[
\rho(r)=\frac{1}{r+\alpha},
\qquad \alpha>0,
\]
which favors nearer landmarks more strongly, or compactly supported profiles that vanish beyond a prescribed distance.

We may also combine a radial profile with ball-dependent prior weights. If $\pi_i>0$ is a prior weight assigned to the ball $C_i$, define
\[
a_i(x) =
\begin{cases}
\pi_i,\rho(d(x,c_i)), & i\in I_x,\\
0, & i\notin I_x.
\end{cases}
\]
Then $\pi_i$ can encode information such as the size of the ball, local sampling density, or an externally chosen importance weight.

In what follows, we write
\[
w(x) = (w_1(x),\dots,w_K(x)) \in \Delta^{K-1}
\]
for the membership vector associated with $x$, where $w_i(x)=\kappa_x({C_i})$. Thus, $w_i(x)\ge 0$,
\[
\sum_{i=1}^K w_i(x)=1,
\]
and
$w_i(x)=0$ whenever $x\notin C_i$. The last condition is the subordination condition. The probabilistic assignment respects the original Ball Mapper cover. Hence, the construction keeps the same cover as classical Ball Mapper, but replaces hard membership by a flexible probability distribution over the balls containing each point.

\subsection{Markov Kernel}

The family of probability measures $\{\kappa_x\}_{x\in X}$ can be viewed as a probabilistic map from data points to covering balls. Instead of assigning a point $x$ to a single ball, or to all balls containing it in a purely binary way, we let $x$ determine a probability distribution on the cover. This is formalized by the notion of a Markov kernel \cite{Kallenberg2021}.

\begin{definition}
Let $A$ and $B$ be finite sets. A \emph{Markov kernel} from $A$ to $B$ is a function
\[
K:A\times 2^B\to[0,1]
\]
such that, for each fixed $a\in A$, the map
\[
E\mapsto K(a,E)
\]
is a probability measure on $B$.
\end{definition}

In the present setting, we define
\[
\kappa:X\times 2^{\mathcal C_\varepsilon(X)}\to[0,1]
\]
by $\kappa(x,A)=\kappa_x(A)$ where $A\subseteq \mathcal C_\varepsilon(X)$.

\begin{proposition}
The map $\kappa$ is a Markov kernel from $X$ to $\mathcal C_\varepsilon(X)$.
\end{proposition}

\begin{proof}
For each fixed $x\in X$, the map $A\mapsto \kappa(x,A)$ is exactly the probability measure $\kappa_x$ on $\mathcal C_\varepsilon(X)$. Hence, $\kappa$ is a Markov kernel from $X$ to $\mathcal C_\varepsilon(X)$.
\end{proof}

The probabilistic Ball Mapper assignment can be interpreted as follows. Given a data point $x$, the kernel $\kappa$ describes how to randomly choose a covering ball associated with $x$. The corresponding membership vector is $w(x)$. Equivalently, after choosing an ordering of the data points, the kernel may be represented by the membership matrix
\[
W=(W_{xi})_{x\in X,\ 1\le i\le K},
\]
and $W_{xi}=w_i(x)$. Each row of $W$ is a probability vector. Moreover, because the assignment is subordinate to the Ball Mapper cover, $w_i(x) = 0$ whenever $x\notin C_i$.

Now, let $\mu$ be a probability measure on $X$. In data-driven applications, $\mu$ is often the empirical measure
\[
\mu({x})=\frac{1}{|X|}, \qquad x\in X,
\]
but other choices are possible when the data points have prescribed weights.

The Markov kernel $\kappa$ transports the data distribution $\mu$ to a probability distribution on the covering balls. Let $\nu:2^{\mathcal C_\varepsilon(X)}\to[0,1]$
be a function defined by
\[
\nu(A) = \sum_{x\in X}\kappa(x,A)\mu({x}), \qquad A\subseteq \mathcal C_\varepsilon(X).
\]
The function $\nu$ is the distribution obtained by first sampling a data point $x$ according to $\mu$, and then sampling a covering ball according to $\kappa_x$.

The weight of the vertex corresponding to $C_i$ is
\[
\nu_i := \nu({C_i}) =\sum_{x\in X}\kappa(x,{C_i})\mu({x}) = \sum_{x\in X}w_i(x)\mu({x}).
\]
Here, $\nu_i$ is the average amount of membership mass assigned to the ball $C_i$.

For the uniform membership rule,
\[
\nu_i = \sum_{x\in C_i}\frac{\mu({x})}{|I_x|}.
\]
For a localized radial-basis membership rule with kernel weights $k(x,c_i)$,
\[
\nu_i = \sum_{x\in C_i}\frac{k(x,c_i)}{\sum_{j\in I_x}k(x,c_j)}\mu({x}).
\]
The probabilistic Ball Mapper construction keeps the same cover and the same underlying nerve graph as classical Ball Mapper, but enriches the vertices with probability weights. These weights combine two pieces of information: the data distribution $\mu$ on $X$ and the chosen soft membership rule $\kappa$.

The vertex weights describe how much probability mass is assigned to individual covering balls. To measure how probability mass is shared between pairs of balls, we introduce a soft overlap matrix.

\begin{definition}
Let $w_i(x)=\kappa_x({C_i})$ be the probabilistic membership of $x$ in $C_i$. The \emph{soft overlap matrix} associated with the probabilistic Ball Mapper cover is the matrix
$Q=(Q_{ij})_{i,j=1}^K$ defined by
\[
Q_{ij} = \sum_{x\in X} w_i(x)w_j(x)\mu({x}).
\]
\end{definition}

The entry $Q_{ij}$ has a natural probabilistic interpretation. Sample a point $x$ from $X$ according to $\mu$. Then, conditional on $x$, independently sample two covering balls according to the same distribution $\kappa_x$. The quantity $Q_{ij}$ is the probability that the first sampled ball is $C_i$ and the second sampled ball is $C_j$.

Equivalently, if $D_\mu$ denotes the diagonal matrix with entries $\mu(\{x\})$, then
\[
Q=W^\top D_\mu W.
\]
Hence, $Q$ is symmetric, nonnegative, and positive semidefinite. Moreover,
\[
\sum_{i=1}^K\sum_{j=1}^K Q_{ij} = \sum_{x\in X}\left(\sum_{i=1}^K w_i(x)\right)\left(\sum_{j=1}^K w_j(x)\right)\mu({x}) = 1.
\]
Thus, $Q$ is a probability distribution on ordered pairs of covering balls.

The marginals of this distribution are exactly the vertex weights. Indeed,
\[
\sum_{j=1}^K Q_{ij} =\sum_{j=1}^K\sum_{x\in X} w_i(x)w_j(x)\mu({x}) = \sum_{x\in X}w_i(x)\left(\sum_{j=1}^K w_j(x)\right)\mu({x}) = \nu_i.
\]
Similarly,
\[
\sum_{i=1}^K Q_{ij}=\nu_j.
\]
Therefore, $Q$ refines the vertex distribution $\nu$ by recording how vertex mass is coupled through shared data points.

The diagonal entry
\[
Q_{ii} = \sum_{x\in X}w_i(x)^2\mu({x})
\]
measures the amount of mass that is repeatedly assigned to the same ball $C_i$ when two balls are sampled from the same data point. Large diagonal values indicate concentrated membership, while smaller diagonal values indicate that mass is more spread across several balls.

The probabilistic overlap should be distinguished from the overlap-count weight already present in the original Ball Mapper formulation \cite{dlotko2019ballmapper}. For two distinct balls, the classical count is
\[
  m_{ij} =\sum_{x\in X}\mathbf{1}_{\{x\in C_i\}}\mathbf{1}_{\{x\in C_j\}},
\]
whereas the probabilistic overlap is
\[
  Q_{ij} = \sum_{x\in X}w_i(x)w_j(x)\mu(x).
\]
A point contained in several balls contributes one unit to every relevant classical incidence and overlap count. In the probabilistic construction, the same point first distributes one unit of mass among its admissible balls; its contribution to $Q$ then distributes that mass over ordered pairs of balls. Consequently,
\[
  \sum_{i=1}^{K}\nu_i=1,
  \qquad
  \sum_{i=1}^{K}\sum_{j=1}^{K}Q_{ij}=1,
  \qquad\text{and}\qquad
  \sum_{j=1}^{K}Q_{ij}=\nu_i.
\]
Thus, the novelty is not the mere availability of weighted vertices or edges. It is the mass-conserving point-to-cover normalization and the coherent marginal--joint interpretation of $(\nu,Q)$. The off-diagonal entries of $Q$ may still be used as edge weights whenever a weighted visualization is desired, while the binary nerve records whether the corresponding overlap exists.

\subsection{A Minimal Worked Example}

Let $X=\{x_1,x_2,x_3\}$ carry the empirical distribution $\mu(x_\ell)=1/3$.  Consider two covering balls satisfying
\[
  C_1\cap X=\{x_1,x_2\},
  \qquad\text{and}\qquad
  C_2\cap X=\{x_2,x_3\}.
\]
Under the uniform-on-support rule, the membership matrix is
\[
  W=
  \begin{pmatrix}
    1 & 0\\
    \tfrac12 & \tfrac12\\
    0 & 1
  \end{pmatrix},
  \qquad
  D_\mu=\tfrac13 I_3.
\]
The induced vertex distribution is
\[
  \nu=W^{\mathsf T}
  \begin{pmatrix}
    1/3\\[1mm]1/3\\[1mm]1/3
  \end{pmatrix}
  =
  \begin{pmatrix}
    1/2\\[1mm]1/2
  \end{pmatrix},
\]
and the soft overlap matrix is
\[
  Q=W^{\mathsf T}D_\mu W
  =
  \begin{pmatrix}
    5/12 & 1/12\\[1mm]
    1/12 & 5/12
  \end{pmatrix}.
\]
The entries of $Q$ sum to one, and each row and column sums to the corresponding component of $\nu$. The classical Ball Mapper edge between $C_1$ and $C_2$ is supported by the single point $x_2$, so its classical overlap count is $m_{12}=1$.  In contrast, $Q_{12}=1/12$: the point $x_2$ has first divided its unit mass equally between the two balls, and the product $w_1(x_2)w_2(x_2)=1/4$ is then weighted by $\mu(x_2)=1/3$. This elementary example shows why $Q$ is a joint probability distribution rather than a rescaled copy of the classical overlap count.

\section{Transport-Based Comparison}

The preceding section associates to a Ball Mapper cover three probabilistic objects: the membership vector $w(x)$ of each data point, the induced vertex distribution $\nu$, and the soft overlap matrix $Q$. We now use these objects to compare probabilistic Ball Mapper graphs.

The main difficulty is that two independently constructed Ball Mapper covers need not have the same number of balls (vertices). Even when both graphs summarize similar data, their vertex sets may have different cardinalities and there is generally no canonical vertex-by-vertex correspondence. Therefore, an entrywise comparison of adjacency matrices or overlap matrices is meaningful only in the special case where the two graphs are built on the same ordered cover.

Optimal transport \cite{villani2009optimal} provides a natural way around this problem. Instead of requiring a one-to-one matching of vertices, we regard the vertices as carrying probability mass and compare graphs by transporting mass from the vertices of one graph to the vertices of the other. This allows one vertex to split its mass across several vertices, and it also allows several vertices to merge into one.

\subsection{Vertex Measures}

Let $X$ be a finite metric dataset with Ball Mapper cover $\mathcal C_\varepsilon(X) = \{C_1,\dots,C_K\}$. We can write a \emph{probabilistic Ball Mapper graph} as a tuple
\[
\mathsf G_X=(V_X,E_X,c_X,\nu_X,Q_X),
\]
where
\begin{itemize}
\item $V_X = \{1,\dots,K\}$ is the vertex set,
\item $E_X$ is the Ball Mapper edge set,
\item $c_X(i)=c_i$ is the landmark associated with the vertex $i$,
\item $\nu_X\in\Delta^{K-1}$ is the vertex distribution induced by the Markov kernel $\kappa$ and the data distribution $\mu_X$, and
\item $Q_X\in\mathbb R^{K\times K}$ is the soft overlap matrix.
\end{itemize}

We first compare only the vertex distributions and vertex features. Let $\mathsf G_X$ and $\mathsf G_Y$ be two probabilistic Ball Mapper graphs with vertex sets $V_X$ and $V_Y$. We assign to each vertex a feature vector. For example, for $i\in V_X$, define
\[
\phi_X(i) =\bigl(c_X(i),\overline d_X(i),r_X(i)\bigr),
\]
where $c_X(i)$ is the landmark, $\overline d_X(i)$ is the normalized degree, and $r_X(i)$ is the normalized soft-overlap spread.

The \emph{normalized degree} is
\[
\overline d_X(i) =
\begin{cases}
\dfrac{\deg_X(i)}{|V_X|-1}, & |V_X|>1,\\[0.5em]
0, & |V_X|=1,
\end{cases}
\]
where $\deg_X(i) = \left|\{k\in V_X \mid (i,k)\in E_X\}\right|$. The \emph{normalized soft overlap spread} is
\[
r_X(i) =
\begin{cases}
1-\dfrac{Q_X(i,i)}{\nu_X(i)}, & \nu_X(i)>0,\\[0.3em]
0, & \nu_X(i)=0.
\end{cases}
\]
The quantity $r_X(i)$ measures the fraction of the mass associated with vertex $i$ that is shared with other vertices. A value near zero indicates that the mass assigned to $i$ is concentrated on $i$ itself, while a larger value indicates stronger overlap with other cover elements.

If landmark coordinates are included in the feature vectors, the two datasets must be embedded in a common metric or feature space so that $d(c_X(i),c_Y(j))$ is defined.  When no such common ambient geometry is available, the landmark term should be omitted and the comparison should use common intrinsic descriptors or the relational formulation below. 

Assume that vertex features lie in a common feature space $\Phi$ equipped with a ground distance $d_\Phi$. For example, we may take $d_\Phi(\phi_X(i),\phi_Y(j))^p$ as
\[
\lambda_{\mathrm{geo}}d(c_X(i),c_Y(j))^p + \lambda_{\mathrm{deg}}|\overline d_X(i)-\overline d_Y(j)|^p + \lambda_{\mathrm{ov}}|r_X(i)-r_Y(j)|^p,
\]
where $\lambda_{\mathrm{geo}},\lambda_{\mathrm{deg}},\lambda_{\mathrm{ov}}\ge 0$ are \emph{hyperparameters} (user-chosen weights). The first term compares landmark geometry, the second compares local graph connectivity, and the third compares local probabilistic overlap.

The probabilistic Ball Mapper graph $\mathsf G_X$ then determines a probability measure on the feature space given by
\[
\alpha_X = \sum_{i\in V_X}\nu_X(i)\delta_{\phi_X(i)}.
\]
Similarly,
\[
\alpha_Y =\sum_{j\in V_Y}\nu_Y(j)\delta_{\phi_Y(j)}.
\]
We compare these two measures using the $p$-Wasserstein distance \cite{santambrogio2015optimal}
\[
W_p(\alpha_X,\alpha_Y)^p =\min_{\pi\in\Gamma(\nu_X,\nu_Y)}\sum_{i\in V_X}\sum_{j\in V_Y} d_\Phi(\phi_X(i),\phi_Y(j))^p\pi_{ij}.
\]
The term $\pi_{ij}$ belongs to the set $\Gamma(\nu_X,\nu_Y)$ which consists of \emph{couplings} between the vertex distributions. In other words,
\[
\Gamma(\nu_X,\nu_Y) = \left\{\pi\in\mathbb R_+^{|V_X|\times |V_Y|} \mid \sum_{j\in V_Y}\pi_{ij}=\nu_X(i) \text{ and } \sum_{i\in V_X}\pi_{ij}=\nu_Y(j)\right\}.
\]
The entry $\pi_{ij}$ is the amount of mass transported from vertex $i$ of $\mathsf G_X$ to vertex $j$ of $\mathsf G_Y$. This formulation does not require $|V_X|=|V_Y|$. It also does not require a vertex ordering or a one-to-one correspondence between the two graphs.

Thus, comparing vertex sets of probabilistic Ball Mapper graphs is simply an optimal transport problem between the vertex measures $\alpha_X$ and $\alpha_Y$.

\subsection{Relational Measures}

The Wasserstein distance above compares vertices through their features. However, a graph is not only a collection of vertices. It also contains relational information between pairs of vertices. In Ball Mapper, this relational information may come from the binary adjacency matrix, the shortest path distance, or the soft overlap matrix.

To compare such relational structure without imposing a fixed vertex correspondence, we may use a Gromov--Wasserstein type term \cite{memoli2011gromov}. Let $R_X:V_X\times V_X\to\mathbb R$ and $R_Y:V_Y\times V_Y\to\mathbb R$ be pairwise relation functions. For example, $R_X(i,k)$ may be the shortest path distance between vertices $i$ and $k$, or it may be a quantity derived from the soft overlap matrix $Q_X$.

Given a coupling $\pi \in \Gamma(\nu_X,\nu_Y)$, the \emph{structural cost} associated with $\pi$ is
\[
\sum_{i,k\in V_X}\sum_{j,\ell\in V_Y}\left|R_X(i,k)-R_Y(j,\ell)\right|^p \pi_{ij}\pi_{k\ell}.
\]
This term is small when pairs of vertices in $\mathsf G_X$ are transported to pairs of vertices in $\mathsf G_Y$ with similar internal relations. A fused comparison combines the vertex-feature cost with the relational cost \cite{vayer2019optimal}.  For $\eta\in[0,1]$, we denote the resulting fused Gromov--Wasserstein-type discrepancy by $\operatorname{FGW}_{p,\eta}(\mathsf{G}_X,\mathsf{G}_Y)$.

When $\eta=0$, this reduces to the ordinary Wasserstein comparison of vertex sets. When $\eta = 1$, it compares only the internal relational structures. Intermediate values of $\eta$ balance the effect from both vertex features and graph structure.

In this work, the relation matrix $R_X$ may be chosen according to the aspect of the Ball Mapper graph we want to compare. Possible choices include the shortest path distance
\[
R_X(i,k)=\operatorname{dist}_{E_X}(i,k),
\]
in the Ball Mapper graph, or a soft-overlap dissimilarity derived from $Q_X$. For example, when $\nu_X(i)\nu_X(k)>0$, we may normalize the soft overlap by
\[
S_X(i,k) =\frac{Q_X(i,k)}{\sqrt{\nu_X(i)\nu_X(k)}}
\]
and use a dissimilarity such as $R_X(i,k) = 1 - S_X(i,k)$.

If two Ball Mapper graphs are built on the same ordered cover, then simpler entrywise comparisons are possible. For example, if $\mathcal C_\varepsilon = \{C_1,\dots,C_K\}$ is shared by both graphs, then the vertex distributions $\nu_X,\nu_Y\in\Delta^{K-1}$ may be compared directly, and the matrices $Q_X,Q_Y$ may be compared entrywise.

The transport formulation above contains this special case as a limiting or constrained situation. We may restrict attention to couplings that concentrate most of their mass near corresponding vertices, or, in the strongest case, to the identity coupling. Without such a prescribed correspondence, the Wasserstein or fused Gromov--Wasserstein formulation is more appropriate.

Entrywise distances, such as $|Q_X-Q_Y|_{\mathrm F}$, are useful only when the two graphs are defined on the same ordered vertex set. For independently fitted Ball Mapper graphs, the transport formulation is the natural replacement.

\subsection{Stability under Perturbations}

We now discuss how the probabilistic Ball Mapper graph changes when the data are perturbed. Let $X= \{x_1,\dots,x_n\}$ and $X'= \{x'_1,\dots,x'_n\}$ be finite subsets of a common metric space. We say that $X'$ is a \emph{$\delta$-perturbation} of $X$ if
\[
d(x_\ell,x'_\ell) \le \delta
\]
for every $\ell=1,\dots,n$.

First, suppose that both datasets are fitted using the same fixed cover
\[
\mathcal C_\varepsilon= \{C_1,\dots,C_K\}
\]
with the same landmarks. Then the vertex set is fixed, and the only quantities that change are the membership vectors $w(x_\ell)$ and $w(x'_\ell)$.

The support of a subordinate membership rule can change when a point crosses the boundary of a covering ball. To isolate this effect, we define the \emph{$\delta$-boundary band}
\[
B_\delta = \left\{x\in X \mid \min_{1\le i\le K} |d(x,c_i) - \varepsilon| \le \delta \right\}.
\]
This is the set of points lying within distance $\delta$ from the boundary of some covering ball. Let $b_\delta=\mu_X(B_\delta)$ be its mass.

We assume that the membership map is \emph{Lipschitz} on each region where the index set $I_x$ is fixed. That is, we assume that there exists $L_\kappa>0$ such that
\[
\|w(x)-w(y)\|_1 \le L_\kappa d(x,y)
\]
whenever $I_x = I_y$. For localized radial-basis membership rules, this assumption is natural away from boundary crossings. Indeed, on any region where the active set $I_x$ is fixed, the same balls receive positive mass, and only the distances $d(x,c_i)$ vary with $x$. 

If the radial profile $\rho$ is Lipschitz continuous on $[0,\varepsilon]$, then
\[
\|\rho(d(x,c_i))-\rho(d(y,c_i))\| \leq L_\rho \|d(x,c_i)-d(y,c_i)\| \leq L_\rho d(x,y),
\]
where the last inequality follows from the reverse triangle inequality. Moreover, if the normalizing denominator $Z_x$ is bounded below by a positive constant on this region, then the normalized weights $w_i(x)$ also vary at most linearly with $x$. Thus, away from changes in the support set $I_x$, localized radial basis assignments satisfy the Lipschitz continuity condition above.

If $x_\ell \notin B_\delta$, then the perturbation does not change the set of balls containing $x_\ell$. Hence, $I_{x_\ell}=I_{x'_\ell}$ and therefore
\[
|w(x_\ell)-w(x'_\ell)|_1 \le L_\kappa\delta.
\]
For points in $B_\delta$, the support may change. Since both $w(x_\ell)$ and $w(x'_\ell)$ are probability vectors, we always have the crude bound
\[
\|w(x_\ell)-w(x'_\ell)\|_1\le 2.
\]
It follows that the vertex distributions satisfy
\[
\|\nu_X-\nu_{X'}\|_1 \le L_\kappa\delta +2b_\delta.
\]
Indeed,
\[
\nu_X-\nu_{X'} =\sum_{\ell=1}^n\left(w(x_\ell)-w(x'_\ell)\right)\mu_X({x_\ell}),
\]
and the preceding pointwise estimates give the bound.

The soft overlap matrices satisfy a similar estimate. Since
\[
Q_X =\sum_{\ell=1}^nw(x_\ell)w(x_\ell)^\top\mu_X({x_\ell}),
\]
and similarly for $Q_{X'}$, we use
\[
\|uu^\top-vv^\top\|_{\mathrm F} \le 2\|u-v\|_1
\]
for probability vectors $u$ and $v$. Therefore,
\[
\|Q_X-Q_{X'}\|_{\mathrm F} \le 2L_\kappa\delta + 4b_\delta.
\]

Thus, in the fixed cover setting, the probabilistic quantities $\nu$ and $Q$ vary continuously away from boundary crossings, with an additional error term controlled by the amount of data mass near ball boundaries.

This should be contrasted with the binary Ball Mapper adjacency matrix. A single point crossing into or out of an overlap region can create or destroy an edge. Consequently, the hard adjacency matrix can change abruptly under small perturbations. The soft overlap matrix $Q$ records the amount of shared membership mass and is therefore better suited for stable quantitative comparison.

If the Ball Mapper cover is recomputed after perturbing the data, then an additional source of variation appears: the landmarks themselves may change. Let $N_X = \{c_X(i) \mid i\in V_X\}$ and $N_{X'}=\{c_{X'}(j) \mid j\in V_{X'}\}$ be the landmark sets of the two independently fitted covers. Their \emph{Hausdorff distance} is
\[
d_{\mathrm{Haus}}(N_X,N_{X'}) = \max\left\{\sup_{a\in N_X}\inf_{b\in N_{X'}}d(a,b), \sup_{b\in N_{X'}}\inf_{a\in N_X}d(a,b)\right\}.
\]
Let $r_\delta=d_{\mathrm{Haus}}(N_X,N_{X'})$. This term measures how much the landmark sets move when the data are perturbed and the cover is rebuilt.

Under additional regularity assumptions, this discussion suggests a stability estimate of the form
\[
\operatorname{FGW}_{p,\eta}(\mathsf G_X,\mathsf G_{X'}) \le C_1\delta + C_2b_\delta + C_3r_\delta.
\]
Here, the term $C_1\delta$ reflects smooth variation of the membership weights away from cover boundaries, $C_2b_\delta$ accounts for the mass of points whose support may change by crossing a ball boundary, and $C_3r_\delta$ accounts for changes in the landmark set when the cover is recomputed. This estimate should be interpreted as a stability principle rather than an unconditional theorem. A formal bound requires Lipschitz continuity assumptions on the membership rule, vertex features, and relation matrices.

\section{Conclusions}

We developed a probabilistic interpretation of Ball Mapper that keeps its metric-ball cover and nerve construction while replacing binary point-to-ball incidence by a probability distribution supported on the balls containing each point.  The resulting membership map is both a partition of unity subordinate to the cover and a Markov kernel from data points to cover elements. The uniform-on-support and localized radial-basis rules illustrate how this formulation can retain Ball Mapper locality while incorporating different degrees of geometric information.

The kernel induces two complementary cover-level objects. The pushforward $\nu$ records how the data distribution is allocated among the vertices, while $Q=W^{\mathsf T}D_\mu W$ records the joint law of two conditionally independent cover assignments generated from the same data point. In particular, $Q$ is symmetric, nonnegative, positive semidefinite, and has marginals $\nu$. This clarifies the distinction from the overlap-count weights already available in classical Ball Mapper: the contribution here is a conserved probability mass and a coherent marginal--joint representation, not the mere introduction of weighted edges.

These probability objects support two levels of graph comparison. When two datasets use a common ordered cover, their vertex distributions and overlap matrices can be compared directly. When covers are fitted independently and have different vertex sets, transport couplings provide a correspondence-free way to compare vertex mass and graph relations. Landmark geometry may be included only when the datasets share an ambient metric or feature space; in the general relational setting, the fused construction should be regarded as a transport-based discrepancy unless its relation functions satisfy additional conditions that establish metric properties.

The stability results have a similarly precise scope.  For a fixed cover, the derived bounds separate smooth variation of the membership weights from the mass $b_\delta$ of points close to ball boundaries. They explain both why $\nu$ and $Q$ can vary continuously away from boundary crossings and why abrupt changes remain possible when supports change. If the cover is recomputed, the landmarks themselves may move or change in number. The additional Hausdorff-type landmark term identifies this source of variation, but the corresponding fused transport estimate remains a stability principle requiring further regularity assumptions rather than a theorem established here.

Several limitations therefore remain.  A Gaussian profile restricted to $d(x,c_i)<\varepsilon$ need not vanish at the ball boundary, so subordination alone does not remove boundary discontinuities. Compactly supported profiles that vanish at $\varepsilon$ may improve regularity. A complete analysis of recomputed $\varepsilon$-nets must also account for landmark selection and for changes in graph relations, while statistical consistency, computational scaling, multiscale parameter selection, and empirical sensitivity to the membership profile require separate study. These questions, together with validation on substantive applications, form the next stage of the work. The present theory supplies the probability and transport framework on which those developments can be based.

\section*{Code availability}

An open source reference implementation of Probabilistic Ball Mapper is available at \url{https://github.com/jhnrckmnznrs/probabilistic_ball_mapper}. The repository provides implementations of the membership rules and graph comparison procedures developed in this manuscript. The present version of the manuscript focuses on the mathematical framework. Application-specific analyses are outside its scope.

\printbibliography

\end{document}